\documentclass[11pt]{article}
\usepackage[utf8]{inputenc}
\usepackage[T1]{fontenc}

\usepackage{geometry}
\usepackage{amssymb,amsmath}
 \usepackage{pstricks, pst-coil, pst-node, pst-tree, multido}
 \usepackage[english]{babel}

\usepackage{graphics}
\usepackage{graphicx}
 \newtheorem{DE}{Definition}[section]

\newcommand {\sm} {\setminus}

\usepackage{latexsym} 
\usepackage{theorem} 
 \newcommand{\qed}{\relax\ifmmode\hskip2em\Box\else\unskip\nobreak\hfill$\Box$\fi}

\newtheorem{theorem}[DE]{Theorem}
\newtheorem{lemma}[DE]{Lemma}

{\theoremstyle{break}\theorembodyfont{\rmfamily}}
{\theoremstyle{break}\theorembodyfont{\rmfamily}}

\newcounter{claim}
\newenvironment{proof}[1][]%
	{\noindent {\setcounter{claim}{0}\it Proof. }{#1}{}}{\qed\vspace{2ex}}
	{\refstepcounter{claim}\vspace{1ex}\noindent {(\it\arabic{claim}) {#1}{}}\it}{\vspace{1ex}}
	{\noindent {}{#1}{}}{ This proves~(\arabic{claim}).\vspace{1ex}}

\DeclareMathOperator{\rw}{rw}
\DeclareMathOperator{\rk}{rank}

\title{A class of graphs with large rankwidth}

\author{Ch\'inh T. Ho\`ang\thanks{Department of Physics and Computer
    Science, Wilfrid Laurier University, Waterloo, Ontario, Canada.}~~and~Nicolas Trotignon\thanks{Univ Lyon, EnsL, UCBL, CNRS,  LIP,
    F-69342, LYON Cedex 07, France. Partially supported by the
     LABEX MILYON (ANR-10-LABX-0070) of Universit\'e de Lyon, within
     the program ‘‘Investissements d'Avenir’’ (ANR-11-IDEX-0007)
     operated by the French National Research Agency (ANR) and by Agence
    Nationale de la Recherche (France) under research grant ANR
    DIGRAPHS ANR-19-CE48-0013-01.}}

\begin{document}

\maketitle

\begin{abstract}
  We describe several graphs with arbitrarily large rankwidth (or
  equivalently with arbitrarily large cliquewidth).  Korpelainen,
  Lozin, and Mayhill [Split permutation graphs, {\em Graphs and
    Combinatorics}, 30(3):633--646, 2014] proved that there exist
  split graphs with Dilworth number~2 with arbitrarily large
  rankwidth, but without explicitly constructing them. We provide an
  explicit construction.  Maffray, Penev, and Vu\v skovi\'c [Coloring
  rings, {\it Journal of Graph Theory} 96(4):642-683, 2021] proved
  that graphs that they call rings on $n$ sets can be colored in
  polynomial time. We show that for every fixed integer
  $n\geq 3$, there exist rings on $n$ sets with arbitrarily large
  rankwidth.  When $n\geq 5$ and $n$ is odd, this provides a new
  construction of even-hole-free graphs with arbitrarily large
  rankwidth.
\end{abstract}

\section{Introduction}

The \emph{cliquewidth} of a graph is an integer intended to measure
how complex the graph is.  It was defined by Courcelle, Engelfriet and
Rozenberg in~\cite{CER-93} and is successful in the sense that many
hard problems on graphs become tractable on graph classes of bounded
cliquewidth~\cite{CMR00}.  This includes for instance finding the
largest clique or independent set, and deciding if a colouring with at
most $k$ colors exists (for fixed $k\in \mathbb N$).  This makes
cliquewidth particularly interesting in the study of algorithmic
properties of hereditary graph classes.

The notion of \emph{rankwidth} was defined by Oum and Seymour
in~\cite{OS-rw}, where they use it for an approximation algorithm for
cliquewidth. They also show that rankwidth and cliquewidth are
equivalent, in the sense that a graph class has bounded rankwidth if,
and only if, it has bounded cliquewidth.  In the rest of this article,
we only use rankwidth, and we therefore refer to results in the
literature with this notion, even if in the original papers, the notion of
cliquewidth is used (recall the two notions are equivalent as long as
we care only about a class being bounded by the parameter or not).

Determining whether a given class of graphs has bounded rankwidth has
been well studied lately, and let us survey the main results in this
direction. For every class of graphs defined by forbidding one or two
graphs as induced subgraphs, it is known whether the class has bounded
or unbounded rankwidth, apart for a very small number of open cases,
see~\cite{DBLP:journals/algorithmica/BonamyBDJPP21} for the most recent
results and~\cite{DBLP:journals/corr/abs-1901-00335} for a survey.

Similar results were obtained for chordal
graphs~\cite{DBLP:journals/jgt/BrandstadtDHP17} and split
graphs~\cite{DBLP:journals/dam/BrandstadtDHP16}. Recall that a
\emph{chordal graph} is a graph such that every cycle of length at
least~4 has a chord, and a \emph{split graph} is a graph whose vertex
set can be partitioned into a clique and a stable set.

A graph is \emph{even-hole-free} if every cycle of even length has a
chord. Determining whether several subclasses of even-hole-free graphs
have bounded rankwidth also attracted some
attention. See~\cite{DBLP:journals/jgt/CameronCH18,DBLP:journals/dm/CameronSHV18,PhDLe17,DBLP:journals/dam/SilvaSS10}
for subclasses of bounded rankwidth
and~\cite{adlerLMRTV:rwehf,DBLP:journals/jgt/SintiariT21} for
subclasses with unbounded rankwidth.

When $G$ is a graph and $x$ a vertex of $G$, we denote by $N(x)$ the
set of all neighbors of $x$.  We set $N[x]=N(x)\cup \{x\}$.  The
\emph{Dilworth number} of a graph $G$ is the maximum number of
vertices in a set $D$ such that for all distinct $x, y\in D$, the two
sets $N(x)\setminus N[y]$ and $N(y)\setminus N[x]$ are non-empty.
In~\cite{DBLP:journals/gc/KorpelainenLM14}, it is proved that split graphs
with Dilworth number~2 and arbitrarily large rankwidth exist. It
should be pointed out that only their
existence is proved, no explicit construction is given.

For an integer $n\geq 3$, a \emph{ring on $n$ sets} is a graph $G$
whose vertex set can be partitioned into $n$ cliques
$X_1, \dots, X_n$, with three additional properties:
\begin{itemize}
\item For all $i\in \{1, \dots, n\}$ and all $x, x'\in X_i$, either
  $N[x]\subseteq N[x']$ or $N[x']\subseteq N[x]$.
\item For all $i\in \{1, \dots, n\}$ and all $x\in X_i$, $N(x)
  \subseteq X_{i-1} \cup X_i \cup X_{i+1}$ (where the addition of
  subscripts is modulo $n$). 
\item For all $i\in \{1, \dots, n\}$, there exists a vertex $x\in X_i$
  that is adjacent to all vertices of $X_{i-1}\cup X_{i+1}$.  
\end{itemize}

Rings were studied in~\cite{maffray2019coloring}, where a polynomial
time algorithm to color them is given.  The notion of ring
in~\cite{maffray2019coloring} is slightly more restricted than ours
(at least 4 sets are required), but it makes no essential difference here.
Observe that the Dilworth number of a ring on $n$ sets is at most
$n$. As explained in~\cite{PhDLe17}, a construction
from~\cite{DBLP:journals/ijfcs/GolumbicR00} shows that there exist
rings with arbitrarily large rankwidth. Also
in~\cite{DBLP:journals/ejc/KwonPS20}, it is proved that the so-called
\emph{twisted chain graphs}, that are similar in some respect to rings
on 3 sets, have unbounded rankwidth.  However, for any fixed integer
$n$, it is not known whether there exist rings on $n$ sets with
arbitrarily large rankwidth.

Our main result is a new way to build graphs with arbitrarily large
rankwidth. The construction has some flexibility, so it allows us to
reach several goals. First, we give split graphs with Dilworth
number~2 with arbitrarily large rankwidth, and we describe them
explicitly. By ``tuning'' the construction differently, we will show that for each
integer $n\geq 3$, there exist rings on $n$ sets with arbitrarily large
rankwidth.  For odd integers $n$, this provides new even-hole-free
graphs with arbitrarily large rankwidth. It should be pointed out that
our construction does not rely on modifying a grid (a classical method
to obtain graphs with arbitrarily large rankwidth).

In Section~\ref{s:rw}, we recall the definition of rankwidth together
with lemmas related to it. In Section~\ref{s:mac}, we give the main
ingredients needed to construct our graphs, called \emph{carousels}, to be
defined in Section~\ref{s:carousel}.  In Section~\ref{s:sketch}, we
give an overview of the proof that carousels have unbounded
rankwidth.  In Section~\ref{s:propagate}, we give several technical
lemmas about the rank of matrices that arise from partitions of the
vertices in  carousels.  In Sections~\ref{s:evenC} and~\ref{s:oddC},
we prove that carousels have unbounded rankwidth (we need two sections
because there are two kinds of carousels, the even ones and the odd
ones).  In Section~\ref{s:app}, we show how to tune carousels in order
to obtain split graphs with Dilworth number~2 or rings.  We
conclude the paper with open questions  in Section~\ref{s:open}.

\section{Rankwidth}
\label{s:rw}

When $G$ is a graph and $(Y, Z)$ a partition of some subset of $V(G)$,
we denote by $M_{G, Y, Z}$ the matrix $M$ whose rows are indexed by
$Y$, whose columns are indexed by $Z$ and such $M_{y, z} = 1$ if
$yz\in E(G)$ and $M_{y, z} = 0$ if $yz\notin E(G)$.  We define
$\rk_G(Y, Z) = \rk(M_{G, Y, Z})$, where the rank is computed on the
binary field.  When the context is clear, we may
refer to $\rk_G({Y, Z})$ as the \emph{rank} of $(Y, Z)$.

A \emph{tree} is a connected acyclic graph. A \emph{leaf} of a tree
is a node incident to exactly one edge. For a tree $T$, we let $L(T)$
denote the set of all leaves of $T$.  A tree node that is not a leaf
is called \emph{internal}.  A tree is \emph{cubic}, if it has at least
two nodes and every internal node has degree $3$.

A \emph{tree decomposition} of a graph $G$ is a cubic tree $T$, such
that $L(T)=V(G)$.  Note that if $\left|V(G)\right|\leq 1$, then $G$
has no tree decomposition.  For every edge $e\in E(T)$, $T\setminus e$
has two connected components, $Y_e$ and $Z_e$ (that we view as trees).
The \emph{width} of an edge $e\in E(T)$ is defined as
$\rk_G({L(Y_e), L(Z_e)})$.  The \emph{width} of $T$ is the maximum
width over all edges of~$T$. The \emph{rankwidth} of $G$, denoted by
$\rw(G)$, is the minimum integer $k$, such that there is a tree
decomposition of $G$ of width $k$.  If $\left|V(G)\right|\leq 1$, we
let $\rw(G)=0$.

Let $G$ be a graph. A partition $(Y, Z)$ of $V(G)$ is \emph{balanced}
if $$|V(G)|/3 \leq |Y|, |Z| \leq 2|V(G)|/3$$ and \emph{unbalanced}
otherwise.  An edge $e$ of a tree decomposition $T$ is (un)-balanced 
if the partition $(L(Y_e), L(Z_e))$  of $G$ as defined above is (un)-balanced.

\begin{lemma}\label{lem:balanced-edge}
  Every tree decomposition $T$ of a graph $G$ has a balanced edge.
\end{lemma}

\begin{proof} 
  For every edge $e \in E(T)$, removing $e$ from $T$ yields two
  components $Y_e$ and $Z_e$. We orient $e$ from $Y_e$ to
  $Z_e$ if $2\left|L(Y_e)\right|<\left|L(Z_e)\right|$.  If there is a
  non-oriented edge $e$, then $e$ is balanced. So, assume that all
  edges are oriented. Since $T$ is a tree, some node $s \in V(T)$ must
  be a sink. Note that $s$ cannot be a leaf.  But $T$ is cubic, so
  each of the three subtrees obtained from $T$ by deleting $s$
  contains less than $\frac13$ of the vertices of $L(T)$, a
  contradiction.
\end{proof}

Interestingly, we do not need the full definition of the rankwidth of
a graph, the following property is enough for our purpose.

\begin{lemma}
  \label{l:rwBalanced}
  Let $G$ be a graph and $r\geq 1$ an integer. If every partition
  $(Y, Z)$ of $V(G)$ with rank less than $r$ is unbalanced, then
  $\rw(G) \geq r$.
\end{lemma}

\begin{proof}
  Suppose for a contradiction that $\rw(G) < r$.  Then there exists a
  tree decomposition $T$ of $G$ with width less than $r$.  Consider a balanced
  edge $e$ of $T$, and  let $Y_e$ and $Z_e$ be the two connected
  components of $T\sm e$.  Then, $(L(Y_e), L(Z_e))$ is a partition of
  $V(G)$ that is balanced and has rank less than $r$, a contradiction
  to our assumptions.
\end{proof}

We do not need many definitions from linear algebra.  In fact the
following basic lemma and the fact that the rank does not increase
when taking submatrices are enough for our purpose. A (0-1)
$n\times n$ matrix $M$ is \emph{diagonal} if $M_{i, j}= 1$ whenever
$i= j$ and $M_{i, j} = 0$ whenever $i\neq j$.  It is
\emph{antidiagonal} if $M_{i, j}= 0$ whenever $i= j$ and
$M_{i, j} = 1$ whenever $i\neq j$.  It is \emph{lower triangular} if
$M_{i, j}= 1$ whenever $i\geq j$ and $M_{i, j} = 0$ whenever $i<j$.

\begin{lemma}
  \label{l:rkTri}
  For every integer $r\geq 1$, the rank of a diagonal, antidiagonal or
  triangular $r\times r$ matrix is at least $r-1$ (in fact it is $r-1$
  when $r$ is odd and the matrix is antidiagonal, and it is $r$
  otherwise).
\end{lemma}

\begin{proof}
  Clear. 
\end{proof}

A (0-1) $n\times n$ matrix $M$ is \emph{near-lower triangular} if
$M_{i, j}= 1$ whenever $i> j$ and $M_{i, j} = 0$ whenever $i<j$.  The
values on the diagonal are not restricted.

\begin{lemma}
  \label{l:nearT}
  For every integer $r\geq 1$, the rank of a near-triangular
  $2r\times 2r$ matrix $M$ is at least $r$.
\end{lemma}

\begin{proof}
  The submatrix $N$ of $M$ formed by the rows of even indexes and the
  columns of odd indexes is a triangular $r\times r$ matrix (formally
  for every $i, j \in \{1, \dots r\}$, $N_{i, j} = M_{2i, 2j-1}$).  By
  Lemma~\ref{l:rkTri}, $\rk(M) \geq \rk(N) = r$.
\end{proof} 

Note that Lemma~\ref{l:nearT} is best possible as shown by the matrix
represented in Figure~\ref{fig:ntm}.

\begin{figure}
$$  \begin{pmatrix}
    1 & 0 & 0 & 0 & 0 & 0\\
    1 & 0 & 0 & 0 & 0 & 0\\
    1 & 1 & 1 & 0 & 0 & 0\\
    1 & 1 & 1 & 0 & 0 & 0\\
    1 & 1 & 1 & 1 & 1 & 0\\
    1 & 1 & 1 & 1 & 1 & 0\\
    \end{pmatrix}
  $$
  \caption{A near-lower triangular $6\times 6$ matrix of rank~3\label{fig:ntm}}
\end{figure}

\section{Matchings, antimatchings and crossings}
\label{s:mac}

In this section, we describe the particular adjacencies that we need to build our
graphs. Suppose that a graph $G$ contains two disjoint ordered sets of
vertices of the same cardinality $k$, say $X=\{u^1, \dots, u^k\}$ and
$X'=\{v^1, \dots, v^k\}$.

\begin{itemize}
\item We say that $(G, X, X')$ is a \emph{regular matching} when:

  for all $j, j'\in \{1, \dots, k\}$, $u^jv^{j'}\in E(G)$ if and only
  if $j=j'$.

\item
  We say that $(G, X, X')$ is a \emph{regular antimatching} when:

  for all $j, j'\in \{1, \dots, k\}$, $u^jv^{j'}\in E(G)$ if and only
  if $j\neq j'$.

\item
  We say that $(G, X, X')$ is a \emph{regular crossing} when:

  for all $j, j'\in \{1, \dots, k\}$, $u^jv^{j'}\in E(G)$ if and only
  if $j+j'\geq k+1$.

\item
  We say that $(G, X, X')$ is a \emph{expanding matching} when:

  for all $j, j'\in \{1, \dots, k\}$, $u^jv^{j'}\in E(G)$ if and only
  if $j' = 2j$ or $j'=2j+1$.

\item
  We say that $(G, X, X')$ is a \emph{expanding antimatching} when:
  
  for all $j, j'\in \{1, \dots, k\}$, $u^jv^{j'}\in E(G)$ if and only
  if $j' \neq 2j$ and $j'\neq 2j+1$.

\item We say that $(G, X, X')$ is a \emph{expanding crossing} when:
  
  for all $j, j'\in \{1, \dots, k\}$, $u^jv^{j'}\in E(G)$ if and only
  if $2j+j'\geq 2k+2$.



\item
  We say that $(G, X, X')$ is a \emph{skew expanding matching} when $k$ equals 2
  modulo 4 and:

  for all $j\in \{1, \dots, (k-2)/4\}$ and $j'\in \{1, \dots, k\}$,
  $u^jv^{j'}\in E(G)$ if and only if $j' = 2j$ or $j'=2j+1$; 

  for all $j\in \{(k-2)/4+1, \dots, (3k+2)/4\}$ and
  $j'\in \{1, \dots, k\}$, $u^jv^{j'}\notin E(G)$; and

  for all $j\in \{(3k+2)/4+1, \dots, k\}$ and $j'\in \{1, \dots, k\}$,
  $u^jv^{j'}\in E(G)$ if and only if $j' = 2j-k-2$ or $j'=2j-k-1$.

\item We say that $(G, X, X')$ is a \emph{skew expanding antimatching} when $k$
  equals 2 modulo 4 and:

  for all $j\in \{1, \dots, (k-2)/4\}$ and $j'\in \{1, \dots, k\}$, 
  $u^jv^{j'}\notin E(G)$ if and only if $j' = 2j$ or $j'=2j+1$; 

  for all $j\in \{(k-2)/4+1, \dots, (3k+2)/4\}$ and $j'\in \{1, \dots, k\}$, 
  $u^jv^{j'}\in E(G)$; and

  for all $j\in \{(3k+2)/4+1, \dots, k\}$ and $j'\in \{1, \dots, k\}$,
  $u^jv^{j'}\notin E(G)$ if and only if $j' = 2j-k-2$ or $j'=2j-k-1$.

\item We say that $(G, X, X')$ is a \emph{skew expanding crossing} when $k$
  equals 2 modulo 4 and:

  for all $j\in \{1, \dots, (k-2)/4\}$ and $j'\in \{1, \dots, k\}$, 
  $u^jv^{j'}\in E(G)$ if and only if $2j+j'\geq k$; 

  for all $j\in \{(k-2)/4+1, \dots, k/2\}$ and $j'\in \{1, \dots, k\}$, 
  $u^jv^{j'}\in E(G)$ if and only if $j'\geq k/2+1$;  

  for all $j\in \{k/2+1, \dots, (3k+2)/4\}$ and $j'\in \{1, \dots, k\}$, 
    $u^jv^{j'}\in E(G)$ if and only if $j'\geq k/2-1$;   and 
  
  for all $j\in \{(3k+2)/4+1, \dots, k\}$ and $j'\in \{1, \dots, k\}$,
  $u^jv^{j'}\in E(G)$ if and only if $2j+j'-2\geq 2k$.
\end{itemize}

A triple $(G, X, X')$ is \emph{regular} if it is a regular matching, a regular
antimatching or a regular crossing, see Figure~\ref{fig:regular}. On
the figures, vertices are partitioned into boxes and several sets receive
names. These will be explained in the next section. 

\begin{figure}
  \begin{center}
    \includegraphics[width=4cm]{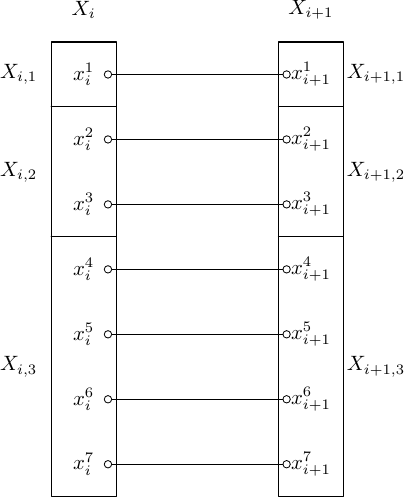}\rule{.5cm}{0cm}
    \includegraphics[width=4cm]{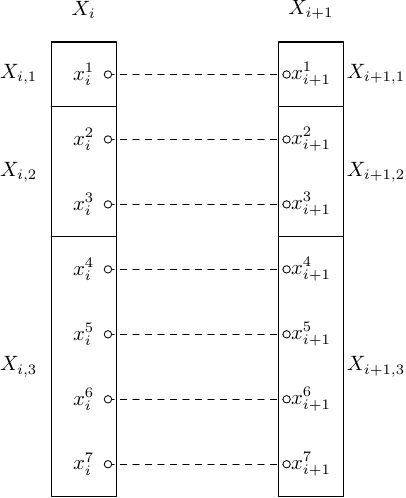}\rule{.5cm}{0cm}
    \includegraphics[width=4cm]{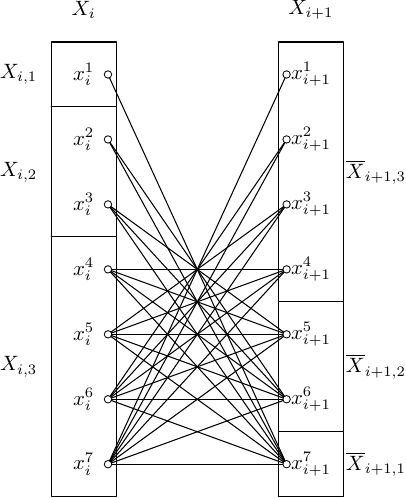}
  \caption{A regular matching, a regular antimatching and a regular
    crossing (on the regular antimatching, only non-edges are
    represented)\label{fig:regular}}
  \end{center}
\end{figure}

A triple $(G, X, X')$ is \emph{expanding} if it is an expanding
matching, an expanding antimatching or an expanding crossing, see
Figure~\ref{fig:expanding}.

\begin{figure}
  \begin{center}
    \includegraphics[width=4cm]{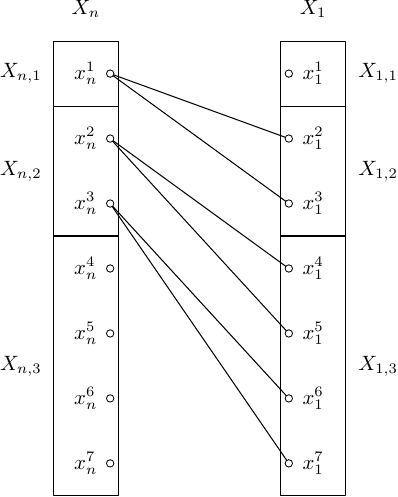}\rule{.5cm}{0cm}
    \includegraphics[width=4cm]{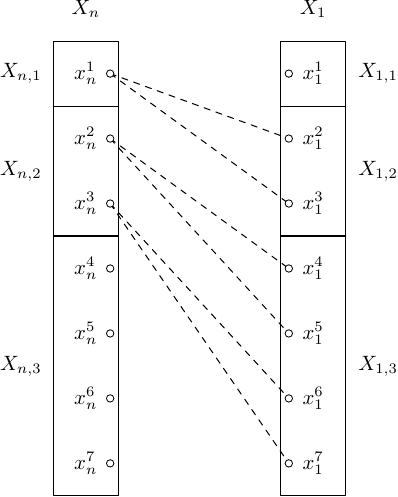}\rule{.5cm}{0cm}
    \includegraphics[width=4cm]{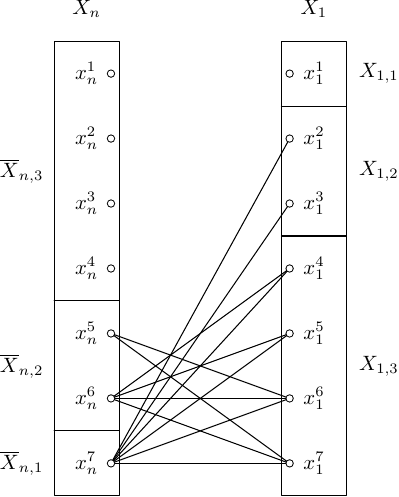}
  \caption{An expanding matching, an expanding antimatching and an expanding
    crossing (on the expanding antimatching, only non-edges are
    represented)\label{fig:expanding}}
  \end{center}
\end{figure}

A triple $(G, X, X')$ is \emph{skew expanding} if it is a skew
expanding matching, a skew expanding antimatching or a skew expanding
crossing, see Figure~\ref{fig:skew}.

\begin{figure}
  \begin{center}
    \includegraphics[width=4cm]{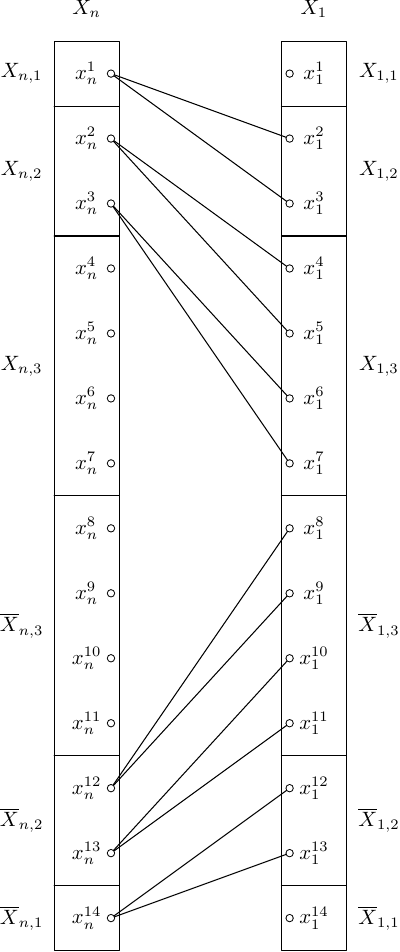}\rule{.5cm}{0cm}
    \includegraphics[width=4cm]{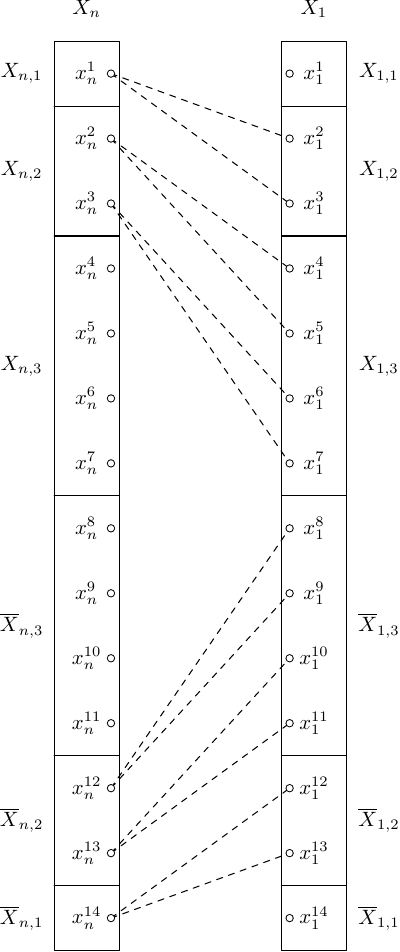}\rule{.5cm}{0cm}
    \includegraphics[width=4cm]{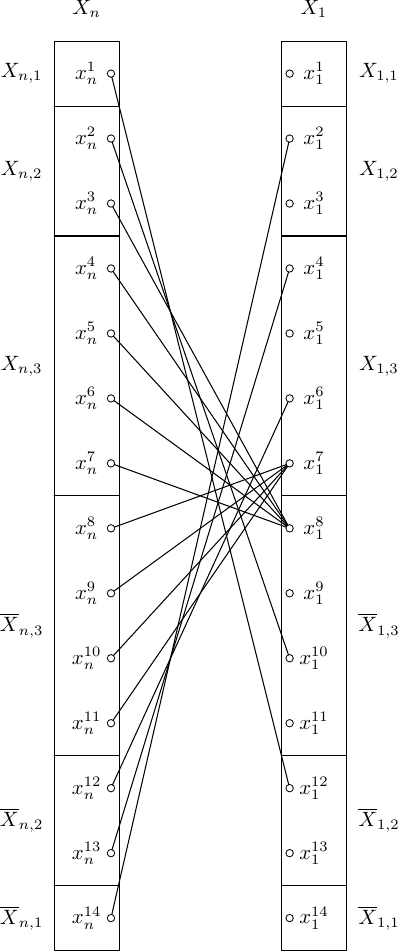}
  \caption{A skew expanding matching, a skew expanding antimatching
    and a skew expanding
    crossing (on the crossing, for each $i\in \{1, \dots, k\}$, only
    the edge  $x_n^ix_1^j$ such that $j$ is minimum is represented and
    on the skew antimatching, only non-edges are
    represented)\label{fig:skew}}
  \end{center}
\end{figure}

A triple $(G, X, X')$ is a \emph{parallel triple} if it is a matching
or an antimatching (regular, expanding or skew expanding).  A triple
$(G, X, X')$ is a \emph{cross triple} if it is a crossing (regular,
expanding or skew expanding).

\section{Carousels}
\label{s:carousel}

The graphs that we construct are called \emph{carousels} and are built
from $n\geq 3$ sets of vertices of equal cardinality $k\geq 1$: $X_1$,
\dots, $X_n$.  So, let $G$ be a graph such that
$V(G) = X_1 \cup \dots \cup X_n$.  Throughout the rest of the paper,
the subscripts for sets $X_i$ are considered modulo $n$.  The graph
$G$ is a \emph{carousel on $n$ sets of cardinality $k$} if:

\begin{enumerate}
\item $(G, X_1, X_2)$ is a regular crossing;
\item for every $i\in \{2, \dots, n-1\}$, $(G, X_i, X_{i+1})$ is a
  regular triple and
\item $(G, X_n, X_1)$ is an expanding triple or a skew expanding triple.
\end{enumerate}

Whether $(G, X_n, X_1)$ should be an expanding triple or a skew
expanding triple depends on how many crossing triples there are among
the triples $(G, X_i, X_{i+1})$ for $i\in \{1, \dots, n\}$.   Let us
explain this. 

Let $s\geq 1$ be an integer. An \emph{even carousel of order $s$ on
  $n$ sets} is a carousel on $n$ sets of cardinality $k$ such that:

\begin{enumerate}
\item $k = 2^s-1$; 
\item the total number of crossing triples among the triples $(G, X_i,
  X_{i+1})$ for $i\in \{1, \dots, n\}$ is even and
\item $(G, X_n, X_1)$ is an expanding triple.
\end{enumerate}

Let $s\geq 1$ be an integer.  An \emph{odd carousel of order $s$ on
  $n$ sets} is a carousel on $n$ sets of cardinality $k$ such that:

\begin{enumerate}
\item $k=2\times (2^s-1)$ (so $k$ is equal to 2 modulo 4); 
\item the total number of crossing triples among the triples $(G, X_i,
  X_{i+1})$ for $i\in \{1, \dots, n\}$ is odd and
\item $(G, X_n, X_1)$ is a skew expanding triple.
\end{enumerate}

In carousels, the edges inside the sets $X_i$ or between sets $X_i$
and $X_j$ such that $j\notin \{i-1, i, i+1\}$ are not specified, they
can be anything.  Our main result is the following.

\begin{theorem}
  \label{th:main}
  For every pair of integers $n\geq 3$ and $r\geq 1$, there exists an
  integer $s$ such that every even carousel and every odd carousel of
  order $s$ on $n$ sets has rankwidth at least~$r$.
\end{theorem}

\section{Outline of the proof}\label{s:sketch}

In this section, we provide an outline of the proof of the main
theorem. The detail will be given in later sections.

For all $i\in \{1,...,n\}$, we let
$X_i = \{x^1_i,...,x^{k}_i\}$.  There is a symmetry in every set $X_i$
and we need some notation for it.  Let $f$ be the function defined for
each integer $j$ by $f(j) = k-j +1$.  We will use a horizontal bar to
denote it as follows.  When $x_i^j$ is a vertex in some set $X_i$, we
denote by $\overline{x}_i^j$ the vertex $x_i^{f(j)}$ and by
$\overline{x}_{i+1}^j$ the vertex $x_{i+1}^{f(j)}$.  We use a similar
notation for sets of vertices: if $S_i\subseteq X_i$ then
$$\overline{S}_i = \{\overline{x}_i^j | x_i^j\in S_i\}$$ and 
$$\overline{S}_{i+1} = \{\overline{x}_{i+1}^j | x_i^j\in S_i\}.$$ Note
that $\overline{\overline{a}} = a$ for any object $a$ such that
$\overline{a}$ is defined.

Each set $X_i$ is partitioned into parts and this differs for the
even and the odd case.

When $G$ is an even carousel, we designate by induction each set $X_i$ as a
\emph{top set} or a \emph{bottom set}. The set $X_1$ is by definition a
top set, and the status of the next ones change at every cross
triple. More formally, for every $i\in \{1, \dots, n-1\}$:

\begin{itemize}
\item If $X_i$ is top set and $(G, X_i, X_{i+1})$ is a parallel triple, then
  $X_{i+1}$ is a top set.
\item If $X_i$ is top set and $(G, X_i, X_{i+1})$ is a cross triple, then
  $X_{i+1}$ is a bottom set.
\item If $X_i$ is bottom set and $(G, X_i, X_{i+1})$ is a parallel triple, then
  $X_{i+1}$ is a bottom set.
\item If $X_i$ is bottom set and $(G, X_i, X_{i+1})$ is a cross triple, then
  $X_{i+1}$ is a top set.
\end{itemize}

Then, for all $i\in \{1, \dots, n\}$ and $j\in \{1, \dots, s\}$, we
set $X_{i, j} = \{x_i^{2^{j-1}}, \dots, x_i^{2^j-1}\}$.  Observe that
for every fixed $i$, the $X_{i, j}$'s (resp.\ the
$\overline{X}_{i, j}$'s) form a partition of $X_i$.  We view every top
set $X_i$ as partitioned by the $X_{i, j}$'s and every bottom set $X_i$
as partitioned by the $\overline{X}_{i, j}$'s.  See Figure~\ref{fig:regular}.

In an odd carousel of order $s$, there are also top sets and bottom sets
(defined as above), but they are all partitioned in the same way: for
all $i\in \{1, \dots, n\}$ and $j\in \{1, \dots, s\}$, we set
$X_{i, j} = \{x_i^{2^{j-1}}, \dots, x_i^{2^j-1}\}$, and the partition
is $(X_{i, 1}, \dots, X_{i, 1}, \overline X_{i, 1}, \dots, \overline
X_{i, s})$. See Figure~\ref{fig:skew}.

To prove Theorem~\ref{th:main}, we fix the integers $n$ and $r$.  We
then compute a large integer $s$ (depending on $n$ and $r$) and we
consider a carousel $G$ of order $s$ on $n$ sets.  We then study the behavior of
a partition $(Y, Z)$ of $V(G)$ of rank less than $r$, our goal
being to prove that it is unbalanced (this proves the rankwidth of $G$
is at least $r$ by Lemma~\ref{l:rwBalanced}).

To check whether a partition is balanced, we need some notation to
measure how many elements of $Y$ some set contains.  For an integer
$m\geq 0$, a set $S\subseteq V(G)$ receives \emph{label $Y_m$} if
$$(m-1)r < |S \cap Y| \leq mr.$$ Note that $S$ has label $Y_0$
if and only if $S\subseteq Z$.  Having label $Y_m$ means that
$\lceil |S\cap Y| / r \rceil = m$.  Observe that for every subset $S$
of $V(G)$, there exists a unique integer $m\geq 0$ such that $S$ has
label $Y_m$.

The first step of the proof is to note that when enumerating the vertices of $X_1$
from $x_1^1$ to $x_1^k$, there are not too many switches
from $Y$ to $Z$ or from $Z$ to $Y$.  Because a large number of changes
would imply that the rank between $X_1$ and $X_2$ is high, this is
formally stated and proved in Lemma~\ref{l:reg4blocks}.  For this to be
true, we need that $(G, X_1, X_2)$ is a crossing, this is why there is
no flexibility in the definition for the adjacency between $X_1$ and
$X_2$ in the definition of carousels.

Since $k$ (the common cardinality of the sets $X_i$) is large enough
by our choice of $s$, we then know that in $X_1$ there must be large
intervals of consecutive vertices in $Y$ or in $Z$, say in $Z$ up to
symmetry.  So, one of the sets that partition $X_1$ is fully in $Z$,
say $X_{1, t}$ for some large integer $t$, and has therefore
label~$Y_0$.  The rest of the proof shows that this label propagates
to the rest of the graph.  By a propagation lemma (stated in the next
section), we prove that $\overline{X}_{2, t}$ (since $(G, X_1, X_2)$
is a cross triple) has a label very close to $Y_0$, namely $Y_0$ or
$Y_1$.  This is because a larger label, say, $Y_m$ with $m>1$, would
give rise to a matrix of rank at least $r$ between $Y$ and $Z$.  And
we continue to apply the propagation mechanism until we reach
$X_{n, 1}$ (or $\overline{X}_{n, 1}$).  The label may be $Y_0$, $Y_1$,
\dots, $Y_{n-1}$, but not larger.

Now, we apply other propagation lemmas to handle the triple
$(G, X_n, X_1)$ (that is expanding or skew expanding by definition of
carousels).  Here the adjacency is designed so that some part that is
twice as large than $X_{i,t}$, namely $X_{1, t+1}$ if $G$ is an even carousel, or
$\overline{X}_{1, t+1}$ if $G$ is an odd carousel, has a label being not too large.

For even carousels, by repeating the procedure above, we prove that
for each $i\in \{1, \dots, n\}$, $X_{i, s-1}$ and $X_{i, s}$ have a
label $Y_m$ with $m$ not too large. And since the size of the sets
$X_{i, j}$ is exponential in $j$, the two sets $X_{i, s-1}$ and
$X_{i, s}$ represent a proportion more than $\frac{3}{4}$ of all the
set $X_i$.  And since $X_{i, s-1}$ and $X_{i, s}$ have label $Y_m$
with $m$ small, they contain mostly vertices from $Z$, so that the
partition $(Y, Z)$ is unbalanced.

For odd carousels, the proof is similar, except that we prove that
$X_{i, s-1}$, $X_{i, s}$, $\overline{X}_{i, s-1}$ and
$\overline{X}_{i, s}$ have label $Y_m$ with $m$ being not too large. These $4n$ sets represent a
proportion more than $\frac{3}{4}$ of all the set $X_i$.

In the next section, we supply the detail of this proof.

\section{Blocks and propagation}
\label{s:propagate}

Throughout the rest of this section, $n\geq 3$, $s\geq 1$ and $G$ is a
an even or an odd carousel of order $s$ on $n$ sets. Also, we assume $r\geq 2$ is
an integer and $(Y, Z)$ is a partition of $V(G)$ of rank less than
$r$.

Suppose $X = \{x^1, \dots, x^k\}$ is an ordered set and $(Y, Z)$ is a
partition of $X$.  We call an \emph{interval} of $X$ any subset of $X$ of
the form $\{x^i, x^{i+1}, \dots, x^j\}$.  A \emph{block of $X$}
w.r.t.\ $(Y, Z)$ is any non-empty interval of $X$ that is fully
contained in $Y$ or in $Z$ and that is maximal w.r.t.\ this property.
Clearly, $X$ is partitioned into its blocks w.r.t.\ $(Y, Z)$.

\begin{lemma}
  \label{l:reg4blocks}
  $X_1$ has fewer than $8r$ blocks w.r.t.\ $(Y, Z)$.
\end{lemma}

\begin{proof}
  Suppose that $X_1$ has at least $8r$ blocks, and let $B_1$, \dots,
  $B_{8r}$ be the $8r$ first ones.  Up to symmetry, we may assume that
  for every $i\in \{1, \dots, 4r\}$, $B_{2i-1}\subseteq Y$ and
  $B_{2i} \subseteq Z$.  For each $i\in \{1, \dots, 4r\}$, we choose
  some element $y_1^{i}\in B_{2i-1}$ and some element
  $z_1^{i}\in B_{2i}$.  We denote by $\overline{y}_2^i$ and
  $\overline{z}_2^i$ the corresponding vertices in $X_2$ of $y_1^i$ and $z_1^i$
  respectively. Formally, $y_1^i = x_1^j$ for some integer
  $j\in \{1, \dots, k\}$,
  $\overline{y}_2^i = \overline{x}_2^j = x_2^{k-j+1}$ and the
  definiton of $\overline{z}_2^i$ is similar.  Let
  $X'_2 = \{\overline{y}_2^1, \overline{z}_2^1, \dots,
  \overline{y}_2^{4r}, \overline{z}_2^{4r}\}$.

  Consider $i, j \in \{1, \dots, 4r\}$, $u\in \{y_1^i, z_1^i\}$ and
  $v\in \{\overline{y}_2^j, \overline{z}_2^j\}$. From the definition
  of regular crossings, we have the following key observation:  

  \begin{itemize}
  \item If $i >j$, then  $uv\notin E(G)$.
  \item If $i < j$, then $uv\in E(G)$.
  \item Note that if $i=j$, the adjacency between $u$ and $v$ is
    not specified, it depends on whether $u$ is $y_1^i$ or $z_1^i$ and
    on whether $v$ is $\overline{y}_2^j$ or $\overline{z}_2^j$.
  \end{itemize}

  Suppose first that $|Y\cap X'_2| \geq |Z \cap X'_2|$.  Then, there
  exist at least $2r$ sets among
  $\{\overline{y}_2^1, \overline{z}_2^1\}$, \dots,
  $\{\overline{y}_2^{4r}, \overline{z}_2^{4r}\}$ that have a non-empty
  intersection with $Y$.  This means that there exist $2r$ distinct
  integers $i_1$, \dots, $i_{2r}$ such that for every
  $j \in \{1, \dots, 2r\}$, some vertex from
  $\{\overline{y}_2^{i_j}, \overline{z}_2^{i_j}\}$ is in $Y$.  We
  denote such a vertex  by $v_j$ and let
  $Y' = \{v_1, \dots, v_{2r}\}$.  We then set
  $Z' = \{z_1^{i_1}, \dots, z_1^{i_{2r}}\}$.
  
  By definition, $Z'\subseteq Z$ and $Y'\subseteq Y$.  Also, by the
  key observation above, the matrix $M_{G[Y' \cup Z'], Y', Z'}$ is a
  near-triangular $2r \times 2r$ matrix.  By Lemma~\ref{l:nearT}, it
  has rank at least $r$.  This proves that $\rk_G(Y, Z) \geq r$, a
  contradiction.
  
  When $|Z\cap X_2| \geq |Y \cap X_2|$, the proof is similar.
\end{proof}

The following lemma explains how labels propagate in regular triples
(represented in Figure~\ref{fig:regular}).

\begin{lemma}
  \label{l:regProp}
  Let $m\geq 0$, $1\leq i < n$ and $1\leq j \leq s$ be integers.

  \begin{enumerate}
  \item Suppose that $(G, X_i, X_{i+1})$ is a regular parallel triple.  If
    $X_{i, j}$ has label $Y_{m}$, then $X_{i+1, j}$ has label $Y_{m'}$
    with $m'\leq m+1$. 

    If $\overline{X}_{i, j}$ has label $Y_{m}$, then
    $\overline{X}_{i+1, j}$ has label $Y_{m'}$
    with $m'\leq m+1$. 

  \item Suppose that $(G, X_i, X_{i+1})$ is a regular cross triple.  If
    $X_{i, j}$ has label $Y_{m}$, then $\overline{X}_{i+1, j}$ has label $Y_{m'}$
    with $m'\leq m+1$. 

    If $\overline{X}_{i, j}$ has label $Y_{m}$, then $X_{i+1, j}$ has label $Y_{m'}$
    with $m'\leq m+1$. 
  \end{enumerate}
\end{lemma}

\begin{proof}
  We first deal with the case when $(G, X_i, X_{i+1})$ is a cross
  triple and $X_{i, j}$ has label $Y_{m}$.  Suppose that the
  conclusion does not hold.  This means that $\overline{X}_{i+1, j}$
  has label $Y_{m'}$ with $m' > m+1$.

  So, $|\overline{X}_{i+1, j} \cap Y| > (m+1)r$.  Since
  $|X_{i, j}\cap Y| \leq mr$, we have
  $|\overline{X}_{i+1, j} \cap Y| - |X_{i, j} \cap Y| \geq r+1$.  So
  there exists a subset $S_i$ of $X_{i, j}$ such that $|S_i| = r+1$,
  $S_i=\{x_i^{k_1}, \dots, x_i^{k_{r+1}} \}$, $S_i \subseteq Z$ and
  $\overline{S}_{i+1} = \{\overline{x}_i^{k_1}, \dots,
  \overline{x}_i^{k_{r+1}} \} \subseteq Y$.  The matrix
  $M_{G[S_i\cup \overline{S}_{i+1}], S_i, \overline{S}_{i+1}}$ is lower
  triangular and has rank at least $r$ by Lemma~\ref{l:rkTri}, a
  contradiction to $(Y, Z)$ being a partition of $V(G)$ of rank less
  than~$r$.

  All the other cases are similar.  When  $\overline{X}_{i, j}$ has
  label $Y_{m}$, the proof is symmetric.  When $(G, X_i, X_{i+1})$ is
  a regular matching, we obtain a diagonal matrix, and when $(G, X_i,
  X_{i+1})$ is a regular antimatching, we obtain a antidiagonal
  matrix. 
\end{proof}

The following lemma describes how labels propagate in expanding triples
(see Figure~\ref{fig:expanding}).

\begin{lemma}
  \label{l:speedProp}
  Let $m\geq 0$ and $j$ be integers such that $1\leq j \leq s-1$.

  \begin{enumerate}
  \item Suppose that $(G, X_n, X_1)$ is a parallel expanding triple.  If
    $X_{n, j}$ has label $Y_{m}$, then $X_{1, j+1}$ has label $Y_{m'}$
    with $m'\leq 2m+2$. 
    
  \item Suppose that $(G, X_n, X_1)$ is a cross expanding triple.  If
    $\overline{X}_{n, j}$ has label $Y_{m}$, then $X_{1, j+1}$ has
    label $Y_{m'}$ with $m'\leq 2m+2$.
  \end{enumerate}
\end{lemma}

\begin{proof}
  We first deal with the case when $(G, X_n, X_{1})$ is a cross
  expanding triple and $\overline{X}_{n, j}$ has label $Y_{m}$.  Suppose
  that the conclusion does not hold.  This means that $X_{1, j+1}$ has
  label $Y_{m'}$ with $m'> 2m+2$.

  So, $|X_{1, j+1} \cap Y| > 2(m+1)r$.  Since
  $|\overline{X}_{n, j} \cap Y| \leq mr$, there exists a subset
  $\overline S_n$ of $\overline X_{n, j}$ such that
  $|\overline S_n| = r+1$, $\overline S_n \subseteq Z$,
  $\overline S_n=\{\overline x_n^{k_1}, \dots, \overline x_n^{k_{r+1}}
  \}$ and for every $\ell \in \{1, \dots, r+1\}$,
  $\{x_1^{2k_\ell}, x_1^{2k_\ell +1}\} \cap Y \neq \emptyset$.  For
  every $\ell \in \{1, \dots, r+1\}$, we choose
  $y^\ell \in Y \cap \{x_1^{2k_\ell}, x_1^{2k_\ell +1}\}$, and set
  $S_1 = \{y^1, \dots, y^{r+1}\}$. The matrix
  $M_{G[\overline S_n\cup {S}_{1}], \overline S_n, {S}_{1}}$ is lower
  triangular and has rank at least $r$ by Lemma~\ref{l:rkTri}, a
  contradiction to $(Y, Z)$ being a partition of $V(G)$ of rank less
  than~$r$.  See Figure~\ref{f:prL63} for an illustration.

  \begin{figure}
    \begin{center}
      \includegraphics[width=4.3cm]{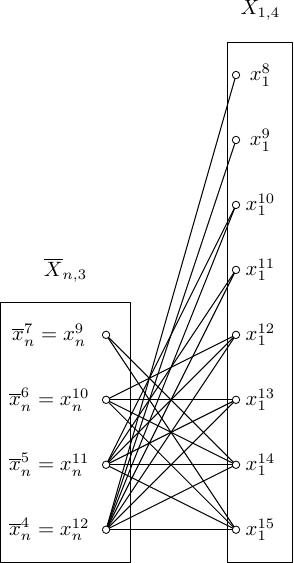}
    \end{center}
    \caption{Proof of Lemma~\ref{l:speedProp} when $k=15$ and $j=3$\label{f:prL63}}
  \end{figure}

  All the other cases are similar.  
  %
  %
  When $(G, X_n, X_1)$ is
  an expanding matching, we obtain a diagonal matrix, and when
  $(G, X_n, X_1)$ is a expanding antimatching, we obtain a
  antidiagonal matrix.
\end{proof}

The following lemma describes how labels propagate in skew expanding
triples (see Figure~\ref{fig:skew}). We omit the proof since it is
similar to the previous one.

\begin{lemma}
  \label{l:skewProp}
  Let $m\geq 0$ and $j$ be integers such that $1\leq j \leq s-1$.

  \begin{enumerate}
  \item Suppose that $(G, X_n, X_{1})$ is a parallel skew expanding
    triple.

    If $X_{n, j}$ has label $Y_{m}$, then $X_{1, j+1}$ has label
    $Y_{m'}$ with $m'\leq 2m+2$.

    If $\overline{X}_{n, j}$ has label $Y_{m}$, then
    $\overline{X}_{1, j+1}$ has label $Y_{m'}$ with $m'\leq 2m+2$.

  \item Suppose that $(G, X_n, X_{1})$ is a skew expanding crossing.


   If $X_{n, j}$ has label $Y_{m}$, then $\overline{X}_{1, j+1}$
   has label $Y_{m'}$ with $m'\leq 2m+2$.

   If $\overline{X}_{n, j}$ has label $Y_{m}$, then $X_{1, j+1}$
   has label $Y_{m'}$ with $m'\leq 2m+2$.

  \end{enumerate}
\end{lemma}

\section{Even carousels}
\label{s:evenC}

In this section, $n\geq 3$ and $r\geq 2$ are fixed integers. We prove
Theorem~\ref{th:main} for even carousels. We therefore look
for an integer $s$ such that every even carousel of
order $s$ on $n$ sets has rankwidth at least~$r$.  We define $s$ as
follows. Let $q$ be an integer such that:

\begin{align}
2^{q+8r-1} &\geq 10r(n+1)(8r+2)2^{8r+2}\label{ineq2}
\end{align}

Clearly $q$ exists.  Let $s=q+8r + 1$.  We now consider an even
carousel $G$ of order $s$ on $n$ sets.  To prove that it has rankwidth at
least~$r$, it is enough by Lemma~\ref{l:rwBalanced} to prove that
every partition $(Y, Z)$ of $V(G)$ with rank less than $r$ is
unbalanced, so let us consider $(Y, Z)$ a partition of $V(G)$ of rank
less than $r$.

\begin{lemma}
  \label{l:ExpandingYzero}
  There exists $t\in \{q, \dots, q + 8r-1\}$ such that $X_{1, t}$ has
  label $Y_0$ or $Z_0$.
\end{lemma}

\begin{proof}
  Otherwise, for every $t\in \{q, \dots, q + 8r - 1\}$, the set
  $X_{1,t}$ is an interval of $X_1$ that must contain elements of $Y$
  and elements of $Z$. Hence, the $8r$ sets $X_{1, q}$, \dots,
  $X_{1, q + 8r - 1}$ show that the interval
  $S= \cup_{j=q}^{q + 8r-1} X_{1, j}$ has at least $8r$ blocks, a
  contradiction to Lemma~\ref{l:reg4blocks}.
\end{proof}

Up to symmetry, we may assume that there exists
$t\in \{q, \dots, q + 8r-1\}$ such that $X_{1, t}$ has label $Y_0$.
We denote by $\widetilde{X}_{i, j}$ the set that is equal to $X_{i, j}$
if $X_i$ is a top set, and that is equal to $\overline{X}_{i, j}$ if
$X_i$ is a bottom set. 

The idea is now to apply Lemmas~\ref{l:regProp} and~\ref{l:speedProp}
repeatedly to show that for all $i\in \{1, \dots, n\}$ and for all
$j\in \{t, \dots, s\}$, the set $\widetilde{X}_{i, j}$ has a label
$Y_m$ with $m$ not too large.  Since $s-t+1 \leq 8r+2$, there are at
most $(n-1)(8r+2)$ applications of Lemma~\ref{l:regProp} that each
augments the index of the label by at most $1$.  There are at most
$8r+2$ applications of Lemma~\ref{l:speedProp} that, in the worst
case, each multiplies the index by $2$ and adds $2$ to it.  In total,
in the worst case, the index is augmented by $1$ at most $(n+1)(8r+2)$
times and multiplied by $2$ at most $8r+2$ times.  Considering that
the additions come first and the multiplications come next will
provide an upper bound to the final result.  The propagation of labels
in an even carousel is illustrated in Figure~\ref{f:propE}. The boxes
represents the sets $X_{i,j}$.  Consider the first box, say box 1, of
$X_1$ that receives label~$Y_0$ by Lemma~\ref{l:ExpandingYzero}. Then,
the block with number 2 has label $Y_0$ or $Y_1$, and so on until all
blocks have a bound on the index of their label (here after nine
steps). We therefore obtain:

\begin{align}
m &\leq (n+1)(8r+2)2^{8r+2}\label{ineqEvenm}
\end{align}

\begin{figure}
  \begin{center}
    \includegraphics[width=4.3cm]{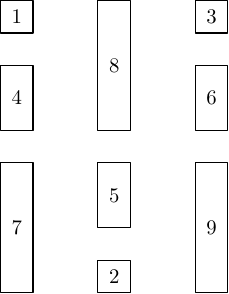}
  \end{center}
  \caption{Propagation of labels in an even carousel\label{f:propE}}
\end{figure}

Let us now prove that $(Y, Z)$ is unbalanced.  To do so, we focus on
the sets $\widetilde{X}_{i, j}$ for $i\in \{1, \dots, n\}$ and
$j\in \{s-1, s\}$.  We denote by $X$ their union.  By elementary
properties of powers of~2, we have $|X| > \frac{3}{4}|V(G)|$.  Also,
by~\eqref{ineqEvenm}, each of the sets $X_{i, j}$ has label $Y_m$ with
$m\leq (n+1)(8r+2)2^{8r+2}$.  Hence, for $i\in \{1, \dots, n\}$ and
$j\in \{s-1, s\}$, $|X_{i, j}\cap Y| \leq r(n+1)(8r+2)2^{8r+2}$.  So,
by inequality~\eqref{ineq2}, the proportion of vertices from $Z$ in
$X$ is at least $\frac{9}{10}$. Hence,
$$|Z| \geq \frac{3}{4} \times \frac{9}{10}|V(G)| = \frac{27}{40}|V(G)|
> \frac{2}{3}|V(G)|.$$

The partition $(Y, Z)$ is therefore unbalanced as claimed.

\section{Odd carousels}
\label{s:oddC}

In this section, $n\geq 3$ and $r\geq 2$ are fixed integers. We prove
Theorem~\ref{th:main} for odd carousels. We therefore look for an
integer $s$ such that every odd carousel of order $s$ on $n$ sets has
rankwidth at least~$r$.  We define $s$ as follows. Let $q$ be an
integer such that:

\begin{align}
2^{q+16r-1} &\geq 10r(n+1)(16r+2)2^{16r+2}\label{ineqSkew2}
\end{align}

Clearly $q$ exists.  Let $s=q+16r + 1$.  We now consider a skew
carousel $G$ of order $s$ on $n$ sets.  To prove that it has rankwidth at
least~$r$, it is enough by Lemma~\ref{l:rwBalanced} to prove that
every partition $(Y, Z)$ of $V(G)$ with rank less than $r$ is
unbalanced, so let us consider $(Y, Z)$ a partition of $V(G)$ of rank
less than $r$.

\begin{lemma}
  \label{l:SkewYzero}
  There exists $t\in \{q, \dots, q + 16r - 2\}$ such that
  $X_{1, t}\cup X_{1, t+1}$ has label $Y_0$ or $Z_0$.
\end{lemma}

\begin{proof}
  Otherwise, for every $i\in \{q, \dots, q+16r-2\}$ the set
  $X_{1,i} \cup X_{1, i+1}$ is an interval of $X_1$ that must contain
  elements of $Y$ and elements of $Z$. Hence, the $8r$ sets
  $X_{1,q} \cup X_{1, q+1}$, $X_{1,q+2} \cup X_{1, q+3}$, \dots,
  $X_{1,q+16r-2} \cup X_{1, q+16r-1}$ show that the interval of $X_1$
  defined as $S= \cup_{j=q}^{q+16r-1} X_{1, j}$ has at least $8r$
  blocks, a contradiction to Lemma~\ref{l:reg4blocks}.
\end{proof}

Up to symmetry, we may assume that there exists
$t\in \{q, \dots, q + 16r-2\}$ such that $X_{1, t}\cup X_{1, t+1}$ has
label $Y_0$.  We denote $\widetilde{X}_{i, j}$ the set that is equal
to $X_{i, j}$ if $X_i$ is a top set, and that is equal to
$\overline{X}_{i, j}$ if $X_i$ is a bottom set.

The idea is now to apply Lemmas~\ref{l:regProp} and~\ref{l:skewProp}
repeatedly to show that all $i\in \{1, \dots, n\}$ and all
$j\in \{t, \dots, s\}$, the sets $X_{i, j}$ and
$\overline{X}_{i, j}$ both have a label $Y_m$ with $m$ not too large.
Since $s-t+1 \leq 16r+1$, there are at most $(n-1)(16r+2)$ applications of
Lemmas~\ref{l:regProp} that each augments the index of the label by at
most $1$.  There are $16r+2$ applications of Lemma~\ref{l:skewProp}
that, in the worst case each adds 2 to the index and multiplies it by
2.  In total, in the worst case, the index is augmented by 1 at most
$(n+1)(16r+2)$ times and multiplied by $2$ at most $16r+2$ times (and
considering that the additions come first and the multiplications come
next provides an upper bound to the final result).  The propagation of
labels in an odd carousel is illustrated in Figure~\ref{f:propO}. The boxes represent the sets $X_{i,j}$ or $\overline{X}_{i,j}$. 
There are boxes 1 and and 1' of $X_1$ that receives label~$Y_0$ by
Lemma~\ref{l:SkewYzero}. Then, the boxes with number 2 and 2' have labels $Y_0$
or $Y_1$, and so on until all boxes have a bound on the index of
their label (here after nine
steps). We therefore obtain:

\begin{align}
m &\leq (n+1)(16r+2)2^{16r+2}\label{ineqOddm}
\end{align} 

 \begin{figure}
    \begin{center}
      \includegraphics[width=4.3cm]{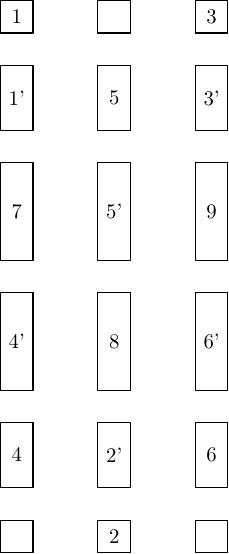}
    \end{center}
    \caption{Propagation of labels in an odd carousel\label{f:propO}}
  \end{figure}

Let us now prove that $(Y, Z)$ is unbalanced.  To do so, we focus on
the $4n$ sets ${X}_{i, j}$, $\overline{X}_{i, j}$ for
$i\in \{1, \dots, n\}$ and $j\in \{s-1, s\}$.  We denote by $X$ their
union.  By elementary properties of powers of~2, we have
$|X| > \frac{3}{4}|V(G)|$.  Also, by~\eqref{ineqOddm}, each of the set $X_{i, j}$ or
$\overline{X}_{i, j}$ has label $Y_m$ with $m\leq (n+1)(16r+2)2^{16r+2}$.
Hence, for $i\in \{1, \dots, n\}$ and $j\in \{s-1, s\}$,
$|X_{i, j}\cap Y| \leq r(n+1)(16r+2)2^{16r+2}$ (and a similar inequality
holds for $\overline{X}_{i, j}$).  So, by inequality~\eqref{ineqSkew2},
the proportion of vertices from $Z$ in $X$ is at least
$\frac{9}{10}$. Hence,
$$|Z| \geq \frac{3}{4} \times \frac{9}{10}|V(G)| = \frac{27}{40}|V(G)|
> \frac{2}{3}|V(G)|.$$

The partition $(Y, Z)$ is therefore unbalanced as claimed.

\section{Applications}
\label{s:app}

In this section, we explain how to tune our carousels  to obtain the results
announced in the introduction.

\subsection*{Split permutation graphs with Dilworth number 2} 

In~\cite{DBLP:journals/gc/KorpelainenLM14}, it is proved that there
exist split permutation graphs with Dilworth number 2 and arbitrarily
large rankwidth.  Theorem~3
in~\cite{DBLP:journals/gc/KorpelainenLM14} states that {\it the class
  of split permutation graphs is precisely the class of split graphs
  of Dilworth number at most 2}. Therefore, it is enough for us to study
split graphs with Dilworth number~2.

We now provide an explicit construction.  Consider an even
carousel on 4 sets $X_1$, $X_2$, $X_3$, $X_4$ and the following
additional properties:

\begin{itemize}
\item $X_1\cup X_3$ is a clique.
\item $X_2 \cup X_4$ is a stable set.
\item The triples $(G, X_1, X_2)$, $(G, X_2, X_3)$, $(G, X_3, X_4)$ 
  are regular  crossings. 
\item The triple $(G, X_4, X_1)$ is an expanding crossing. 
\end{itemize}

It is clear that graphs constructed in this way are split graphs
because $X_1\cup X_3$ is a clique and $X_2\cup X_4$ is a stable set.
To check that they have Dilworth number~2, it is enough to notice that
for all vertices $x, y\in X_1\cup X_2$, either
$N(x)\setminus N[y] =\emptyset$ or $N(y)\setminus N[x] = \emptyset$,
and for all vertices $x, y\in X_3\cup X_4$, either
$N(x)\setminus N[y] =\emptyset$ or $N(y)\setminus N[x] = \emptyset$.

\subsection*{Rings}

We consider a carousel on $n \geq 3 $ sets, all are cliques,
with all triples being crossings. It a ring and it is easy to check
that all its induced cycles have length~$n$. If $n$ is even, this is an
even carousel; and if $n$ is odd, this is an even-hole-free odd
carousel. For a fixed $n$, this gives rings with arbitrarily large
rankwidth.

\subsection*{Flexibility of carousels}

Observe that in this section, we never use matchings and antimatchings,
that were included in the definition for possible later use of carousels. More
generality could be allowed (for instance by not forcing $(G, X_1, X_2)$ to be
crossing, and instead forcing some $(G, X_i, X_{i+1})$ to be
crossing). This might be useful, but we felt that it would make the
explanation too complicated.

\section{Open questions}
\label{s:open}

It would be nice to characterize the class of induced subgraphs of rings on $n$ sets by
forbidden induced subgraphs.  For $n=3$ this seems to be
non-trivial. It might be easier for larger~$n$. On
Figure~\ref{f:nonRing}, a list of obstructions for $n=3$ is given but
we do not know whether it is complete.

\begin{figure}
  \begin{center}
  \parbox[c]{3cm}{\includegraphics[width=.7cm]{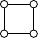}}
  \rule{1cm}{0cm}
  \parbox[c]{3cm}{\includegraphics[width=.7cm]{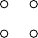}}
  \rule{1cm}{0cm}
  \parbox[c]{3cm}{\includegraphics[width=3cm]{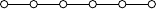}}

  \rule{0cm}{1.5cm}
  
  \parbox[c]{3cm}{\includegraphics[width=1.5cm]{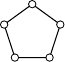}}
  \rule{1cm}{0cm}
  \parbox[c]{3cm}{\includegraphics[width=1.5cm]{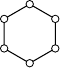}}
  \rule{1cm}{0cm}
  \parbox[c]{3cm}{\includegraphics[width=2cm]{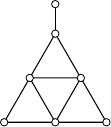}}

  \rule{0cm}{1.5cm}
  
  \parbox[c]{3cm}{\includegraphics[width=4cm]{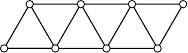}}

  \end{center}
  \caption{Graphs that are not induced subgraphs of rings on 3 sets,
    and that are minimal with respect to this
    property\label{f:nonRing}}
\end{figure}

\begin{figure}
  \begin{center}
  \includegraphics[width=2.5cm]{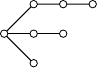}
  \end{center}
  \caption{$S_{1, 2, 3}$\label{f:s123}}
\end{figure}

Let $S_{1, 2, 3}$ be the graph represented on Figure~\ref{f:s123}, and
let $C$ the class of graphs that contain no triangle and no
$S_{1, 2, 3}$ as an induced subgraph. A famous open
question~\cite{DBLP:journals/jcss/DabrowskiDP17} is to determine
whether graphs in $C$ have bounded rankwidth. We tried many ways to
tune our construction, but each time, there was either a triangle or
an $S_{1, 2, 3}$ in it.  So, to the best of our knowledge, our
construction does not help to prove that the rankwidth is unbounded
for this class.  Similarly, whatever tuning we try, it seems that our
construction does not help to solve the open problems
from~\cite{DBLP:journals/dam/BrandstadtDHP16,
  DBLP:journals/jgt/BrandstadtDHP17,
  DBLP:journals/jcss/DabrowskiDP17}. This might be because we did not
try hard enough to tune our construction, or because it is not the
right approach. And of course, the rankwidth might be bounded for all
these classes.

In~\cite{DBLP:journals/order/DaligaultRT10}, a conjecture is given,
stating that in each class of graphs that is closed under taking induced
subgraphs and where the rankwidth is unbounded, there should be an
infinite sequence of graphs $G_1, G_2, \dots$ such that for each pair
of distinct integers $i, j$, $G_i$ is not an induced subgraph of
$G_j$.  This was later disproved,
see~\cite{DBLP:journals/jct/LozinRZ18}, but we wonder whether
carousels may provide  new counter-examples.  This is difficult to
check, because it not clear whether ``big'' carousels contain 
smaller ones as induced subgraphs.

\section{Acknowledgement}

We are grateful to O{-}joung Kwon and Sang{-}il Oum for useful
discussions. We also thank an anonymous referee who found a mistake in
the first version of Lemma~\ref{l:speedProp}.

This work was supported by the Canadian Tri-Council Research Support Fund. The author  C.T.H. was supported by individual NSERC Discovery Grant.


\end{document}